\documentclass[compsoc,conference,letterpaper,10pt,times]{IEEEtran}
\usepackage{url}
\usepackage{graphicx}
\usepackage{amssymb,amsmath,amsfonts,amsthm}

\newtheorem{definition}{Definition}
\newtheorem{lemma}{Lemma}
\newtheorem{theorem}{Theorem}
\newtheorem{corollary}[theorem]{Corollary}

\begin{document}






%

\title{HoneyFaces: Increasing the Security and Privacy\\
of Authentication Using Synthetic Facial Images}
%
%
%
%
%

%


\author{%
\IEEEauthorblockN{Mor Ohana}
\IEEEauthorblockA{Computer Science Dept.\\
University of Haifa\\
Haifa, Israel 31905\\
Email: morohana17@gmail.com
}
\and
\IEEEauthorblockN{Orr Dunkelman}
\IEEEauthorblockA{Computer Science Dept.\\
University of Haifa\\
Haifa, Israel 31905\\
Email: orrd@cs.haifa.ac.il
}
\and
\IEEEauthorblockN{Stuart Gibson}
\IEEEauthorblockA{School of Physical Sciences\\
University of Kent\\
Canterbury CT2 7NF, Kent\\
Email: s.j.gibson@kent.ac.uk
}
\and
\IEEEauthorblockN{Margarita Osadchy}
\IEEEauthorblockA{Computer Science Dept.\\
University of Haifa\\
Haifa, Israel 31905\\
Email: rita@cs.haifa.ac.il
}
}

\maketitle
\begin{abstract}

One of the main challenges faced by Biometric-based authentication systems is the need to offer secure authentication while maintaining the privacy of the biometric data.
Previous solutions, such as Secure Sketch and Fuzzy Extractors, rely on assumptions that cannot be guaranteed in practice, and often affect the authentication accuracy.

In this paper, we introduce HoneyFaces: the concept of adding a large set of synthetic faces (indistinguishable from real) into the biometric ``password file''. This password inflation protects the privacy of users and increases the security of the system without affecting the accuracy of the authentication. In particular, privacy for the real users is provided by  ``hiding'' them among a large number of fake users (as the distributions of synthetic and real faces are equal). In addition to maintaining the authentication accuracy, and thus not affecting the security of the authentication process, HoneyFaces offer several security improvements: increased exfiltration hardness, improved leakage detection, and the ability to use a Two-server setting like in HoneyWords. Finally, HoneyFaces can be combined with other security and privacy mechanisms for biometric data.

We implemented the HoneyFaces system and tested it with a password file composed of 270 real users. The ``password file'' was then inflated to accommodate up to $2^{36.5}$ users (resulting in a 56.6 TB ``password file''). At the same time, the inclusion of additional faces does not affect the true acceptance rate or false acceptance rate which were 93.33\% and 0.01\%, respectively.

\end{abstract}


\section{Introduction}
\label{sec:Intro}

Biometric authentication systems are becoming prevalent in access control and in consumer technology. In such systems, the user submits their user name and his/her biometric sample, which is compared to the stored biometric template associated with this user name (one-to-one matching).\footnote{We note that
one-to-many solutions (where the biometrics is compared to that of all users
in the system) demand significantly larger computational effort, and are used
in systems where there is no user cooperation.} The popularity of biometric-based systems stems from a popular belief that such authentication systems are more secure and user friendly than systems based on passwords. At the same time, the use of such systems raises concerns about the security and privacy of the stored biometric data. Unlike passwords, replacing a compromised biometric trait is impossible, since biometric traits (e.g., face, fingerprint, and iris) are considered to be unique. Therefore, the security of biometric templates is an important issue when considering biometric based systems.  Moreover, poor protection of the biometric templates can have serious privacy implications on the user, as discussed in previous work~\cite{Ratha,Uludag}.


Various solutions have been proposed for protecting biometric templates (e.g,~\cite{Ratha,Uludag}). The most prominent of them are secure
sketch~\cite{FuzzyCommitment} and fuzzy extractors~\cite{FuzzyExtractors}. Unfortunately, these solutions are not well adopted in practice.
The first reason for this is the tradeoff between security and usability due to the
degradation in recognition rates~\cite{manesh10}. The second reason is
related to the use of tokens that are required for storing the helper data, thus
affecting usability. Finally,
these mechanisms rely on assumptions which are hard to
verify (e.g., the privacy guarantees of secure sketch assume that the
biometric trait is processed into an almost full entropy string).


 In this work we propose a different approach for protecting biometric templates called {\em HoneyFaces}. In this approach, we hide the real biometric templates among a very large number of synthetic templates that are  indistinguishable from the real ones. Thus, identifying real users in the system becomes a very difficult `needle in a haystack' problem. At the same time, HoneyFaces does not require the use of tokens nor does it affect recognition rate (compared to a system that does not provide any protection mechanism). Furthermore, it can be integrated with other privacy solutions (e.g., secure sketch), offering additional layers of security and privacy.


For the simplicity of the discussion, let us assume that all biometric templates (real and synthetic) are stored in a \emph{biometric ``password file''}. Our novel approach enables the size of this file to be increased by several orders of magnitudes.
Such inflation offers a 4-tier defense mechanism
for protecting the security and privacy of biometric templates with no
usability overhead. Namely, HoneyFaces:
\begin{itemize}
\item Reduces the risk of the biometric password file leaking;
\item Increases the probability that such a leak is detected online;
\item Allows for post-priori detection of the (biometric) password file leakage;
\item Protects the privacy of the biometrics in the case of leakage;
\end{itemize}
In the following we specify how this mechanism works and its applications in different settings.

The very large size of the ``password file'' 
improves the \textbf{resilience of system against its exfiltration}. We show that one can inflate a system with 270 users (180 KB ``password file'') into a system with up to $2^{36.5}$ users (56.6 TB ``password file''). Obviously, exfiltrating such a huge amount of information is hard. Moreover, by forcing the adversary to leak a significantly larger amount of data (due to the inflated file) he either needs significantly more time, or has much higher chances of being caught by Intrusion Detection Systems. Thus, the file inflation facilitates in \textbf{detecting the leakage} while it happens.

The advantages of increasing the biometric ``password file'' can be demonstrated in
networks whose outgoing bandwidth is very limited, such as air-gap networks (e.g., those
considered in~\cite{Shamir1,Shamir2}). Such networks are usually deployed in high-security restricted areas, and thus are expected to employ biometric authentication, possibly in conjunction with other authentication mechanisms. Once an adversary succeeds in infiltrating the network, he usually has a very limited bandwidth for exfiltration, typically using a physical communication channel of limited capacity (with a typical bandwidth of less than~1 Kbit/sec). In such networks, inflating the size of the database increases the
resilience against exfiltration of the database. Namely, exfiltrating 180 KB of information (the size of a biometric ``password file'' in a system with 270 users) takes a reasonable time even in low bandwidth channels compared with 56.6 TB (the size of the inflated biometric ``password file''), which takes more than 5.2 days for exfiltration in 1 Gbit/sec, 14.4 years in 1 Mbit/sec, or about 14,350
years from an air-gaped network at the speed of 1 Kbit/sec.

Similarly to HoneyWords~\cite{HoneyWords}, the fake accounts enable \textbf{detection of leaked files}. Namely, by using two-server authentication settings, each authentication query is first sent to the server that contains the inflated password file. Once the first server authenticates the user, it sends a query to the second server that contains only the legitimate accounts, thus detecting whether a fake account was invoked with the ``correct'' credentials. This is a clear evidence that despite the hardness of exfiltration, the password file (or a part of it) was leaked.

All the above guarantees heavily rely on the inability of the adversary to isolate the real users from the fake ones. We show that this task is nearly impossible in various adversarial settings (when the adversary has obtained access to the password file). We also show that running  membership queries to identify a real user by matching a facial image from an external source to the biometric ``password file'' is computationally infeasible.\footnote{We alert the reader that for any identification system, once the adversary obtains all the user's credentials (user name, biometric, token, etc.) and the password file, she can easily and efficiently perform a membership query by an attempted login.} We analyze the robustness of the system in the worst case scenario in which the adversary has the facial images of all users except one and he tries to locate the unknown user among the synthetic faces. We show that the system protects the privacy of the users in this case too.
To conclude, HoneyFaces \textbf{protects the biometric templates of real users} in all settings that can be protected.

The addition of a large number of synthetic faces may raise a concern about the degradation of the authentication accuracy. However, we show that this is not the case. The appearance of faces follows a multivariate Gaussian distribution, which we refer to in this article as \emph{face-space}, the parameters of which are learned from a set of real faces, including the faces of the system users. We sample synthetic faces from the same generative model constraining them to be at a certain distance from real and other synthetic faces. We selected this distance to be sufficiently large that new samples of real users would not collide with the synthetic ones. Even though such a constraint limits the number of faces the system could produce, the number remains very large. Using a training set of 500 real faces to build the generative face model, we successfully created $2^{36.5}$ synthetic faces.




\subsection{Problem Statement}
Our HoneyFaces system requires a method for generating synthetic faces which
satisfies three requirements:
\begin{itemize}
\item The system should be able to generate a (very) large number of unique synthetic faces.
\item These synthetic faces should be indistinguishable from real faces.
\item The synthetic faces should not affect the authentication accuracy of real users.
\end{itemize}
These requirements ensure that the faces of the real users can hide among the synthetic ones, without affecting recognition accuracy.

\subsection{Relation to Previous Work}
\label{sec:sub:PreviousWork}
There are two lines of research related to the ideas introduced in this paper. One of them is HoneyObjects, discussed in Section~\ref{sec:sub:sub:HoneyObjects}. The second one, discussed in Section~\ref{sec:sub:sub:biometric_synthesis}, is the synthesis of biometric traits.

\subsection{HoneyObjects}
\label{sec:sub:sub:HoneyObjects}
HoneyObjects are widely used in computer security. The use of honeypot users (fake accounts) is an old trick used by system administrators. Login attempts to such accounts are a strong indication
that the password file has leaked. Later, the concept of
Honeypots and Honeynets was developed \cite{Honeypot}.
These tools are used to lure adversaries
into attacking decoy systems, thus exposing their tools and
strategies.
Honeypots and Honeynets became widely used and deployed in the computer
security world, and play an important role in the mitigation of cyber risks.

Recently, Juels and Rivest introduced HoneyWords~\cite{HoneyWords}, a system
offering decoy passwords in addition to the correct one. A user first
authenticates to the main server using a standard password-based
authentication in which the server can keep track of the number of failed attempts. Once one of the stored passwords is used, the server
passes the query to a second server which stores only the correct
password. Identification of the use of a decoy password by the second server,
suggests that the password file has leaked.
Obviously,
just like in Honeypots and Honeynets, one needs to make sure that the
decoy passwords are sampled from the same space as the real passwords
(or from a space as close as possible). To this end, there is a need to
model passwords correctly, a non-trivial task, which was approached
in several works \cite{HoneyWords,ModelPasswords1,ModelPasswords2}. Interestingly,
we note that modeling human faces
was extensively studied and very good models exist (see the discussion in
Section~\ref{sec:sub:Generating}).\footnote{Our idea of significantly inflating the password file by adding many fake accounts can also be used in the case of non-biometric authentication systems. In addition to the security gain offered by fake accounts, the increased size of the password file makes exfiltration attempts harder.}


In HoneyWords it is a simple matter to change a user's password once if it has been compromised. Clearly it is not practicable to change an individual's facial appearance. Thus, when biometric data is employed, the biometric ``password file'' itself should be protected. HoneyFaces protects the biometric data by inflating the ``password file'' such that it prevents leaks, which is a significant difference between HoneyWords and HoneyFaces.


Another decoy mechanism suggested recently, though not directly related to our work,
is Honey Encryption~\cite{HoneyEncryption}. This is an encryption procedure
which generates ciphertexts that are decrypted to different
(yet plausible) plaintexts when decrypted
under one of a few wrong keys, thus making ciphertext-only exhaustive search
harder.

\subsection{Biometric Traits Synthesis}
\label{sec:sub:sub:biometric_synthesis}
Artificial biometric data are understood as biologically meaningful data for existing biometric systems~\cite{YanushkevichSSWS07}. Biometric data synthesis was suggested for different biometric traits, such as faces (e.g.,~\cite{blanz1999morphable,GibsonSB03,solomon2013interactive,sumi2006study}), fingerprints (e.g, ~\cite{araque2002synthesis,Cappelli15,ImdahlHG15,Kuecken}) and iris (e.g.,~\cite{CuiWHTS04,MakthalR05,ZuoSC07}). The main application of biometrics synthesis has been augmenting training sets and validation of biometric identification and authentication systems (see~\cite{YanushkevichSSWS07} for more information on synthesis of biometrics). Synthetic faces are also used in animation, facial composite construction, and experiments in cognitive psychology.

Making realistic synthetic biometric traits has been the main goal of all these methods.
However, the majority of previous work did not address the question of distinguishing the synthetic samples from the real ones.

The work in iris synthesis~\cite{MakthalR05,ZuoSC07} analyses the quality of artificial samples by clustering synthetic, real, and non-iris images into two clusters iris/non-iris. Such a problem definition is obviously sub-optimal for measuring indistinguishability. Supervised learning using real and synthetic data labels has much better chances of success in separating between real and synthetic samples than unsupervised clustering (a weaker learning algorithm) into iris/non-iris groups. These methods also used recognition experiments, in which they compare the similarity of the associated parameters derived from real and synthetic inputs. Again, this is an indirect comparison that shows the suitability of the generation method for evaluating the quality of the recognition algorithm, but it is not enough for testing the indistinguishability between real and synthetic samples.

In fingerprints, it was shown that synthetic samples generated by different methods could be distinguished from the real ones with high accuracy~\cite{GottschlichH14}. Subsequent methods for synthesis~\cite{ImdahlHG15} showed better robustness against distinguishing attacks that use statistical tests based on~\cite{GottschlichH14}.

 Several methods for synthetic facial image generation~\cite{GibsonSB03,solomon2013interactive} provide near photo-realistic representations, but to the best of our knowledge, the question of indistinguishability between real and synthetic faces has not been addressed before.

\subsection{Organization of the Paper}
\label{sec:sub:Structure}

Section~\ref{sec:sub:Generating} describes, with justification, the method we use for generating HoneyFaces. In Section~\ref{sec:System} we present our setup for employing HoneyFaces in a secure authentication system.
The privacy analysis of HoneyFaces, discussed in Section~\ref{sec:PrivacyAnalysis}, shows that 
the adversary cannot obtain private biometric information from the biometric ``password file''.
Section~\ref{sec:SecurityAnalysis} analyses the additional security offered by inflating the ``password file''. We conclude the paper in Section~\ref{sec:Summary}.

\section{Generating HoneyFaces Images}
\label{sec:sub:Generating}

Biometric systems take a raw sample (usually an image) and process it
to extract features or a representation vector, robust (as much as possible)
to changes in
sampling conditions. In the HoneyFaces system, we have an additional requirement ---
the feature space should allow sampling of artificial ``outcomes'' (faces) in
\emph{large numbers}. These synthetic faces will be used as the passwords of the fake users.


\begin{figure*}[t]
\center{
\includegraphics[width=1.5in]{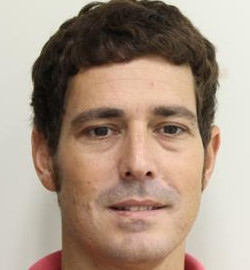}
\includegraphics[width=1.5in]{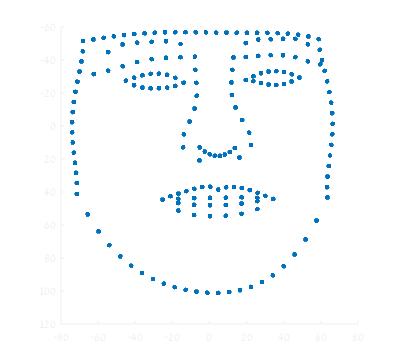}
\includegraphics[width=1.5in]{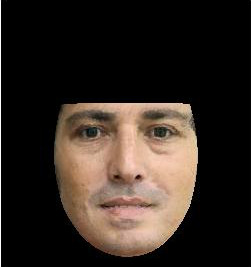}
\caption {Shape vector $x$ (center) and the shape free texture vector $g$ (right) used to obtain the AM coefficient of a facial image (left).} \label{fig_AMM_example}}
\end{figure*}

Different models have been proposed for generating and representing faces including, active appearance models~\cite{CootesET01}, 3D deformable models \cite{blanz1999morphable}, and convolutional neural networks~\cite{li2016convolutional,zhang2015end}. Such models have been used in face recognition, computer animation, facial composite construction (an application in law enforcement), and experiments in cognitive psychology.

Among these models we choose the active appearance model~\cite{CootesET01} for implementing the HoneyFaces concept. An active appearance model is a parametric statistical model that encodes facial variation, extracted from images, with respect to a mean face. This work has been extended and improved in many subsequent papers (e.g.,~\cite{Matthews_2004_4601,TzimiropoulosP13,WuLD08}).  In this context the word `active' refers to fitting the appearance model (AM) to an unknown face to subsequently achieve automatic face recognition~\cite{EdwardsCT98}. AM  can also be used with random number generation to create plausible, yet completely synthetic, faces. These models achieve near photo-realistic representations that preserve identity, although are less effective at modeling hair and finer details, such as birth marks, scars, or wrinkles which exhibit little or no spatial correspondence between individuals.

Our choice of using the AM for HoneyFaces is motivated by two reasons: 1) The representation of faces within an AM is consistent with human visual perception and hence also consistent with the notion of face-space~\cite{valentine1991unified}.  In particular, perceptual similarity of faces is correlated with distance in AM space~\cite{Lewis04}. 2) AM is a well understood model used previously in face synthesis (e.g.~\cite{GibsonSB03,solomon2013interactive}).

Alternative face models may also be considered, provided a sufficient number of training images (as the functions of the representation length) is available to adequately model the facial variation within the population of real faces.
Recent face recognition technology uses deep learning (DL) methods as they provide very good representations for verification.  However, the image reconstruction quality from DL representation is still far from being satisfactory for our application.

%

\subsection{Face representation using AM coefficients}\label{sec:sub:coefficients}
\begin{figure*}[t]
\center{
\includegraphics[width=1.5in]{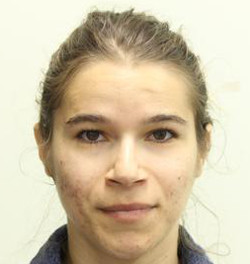}
\includegraphics[width=1.5in]{face_real_85_1.jpg}
\includegraphics[width=1.5in]{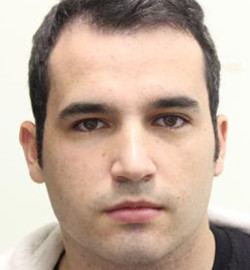}
\includegraphics[width=1.5in]{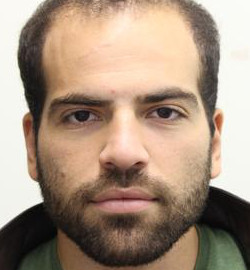}\\
\includegraphics[width=1.5in]{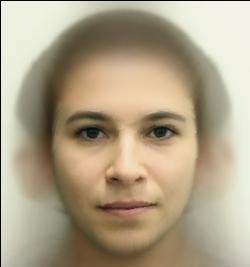}
\includegraphics[width=1.5in]{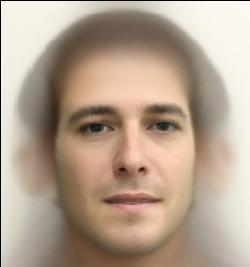}
\includegraphics[width=1.5in]{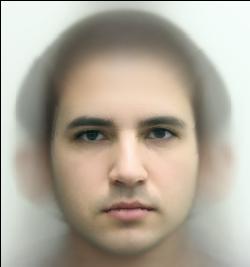}
\includegraphics[width=1.5in]{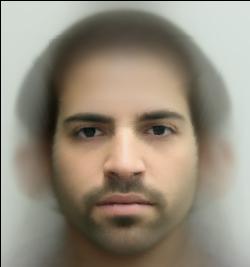}\\
\includegraphics[width=1.5in]{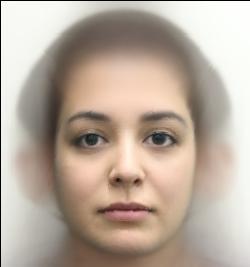}
\includegraphics[width=1.5in]{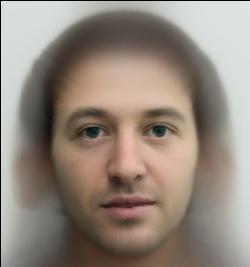}
\includegraphics[width=1.5in]{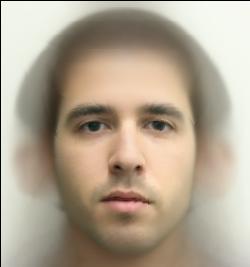}
\includegraphics[width=1.5in]{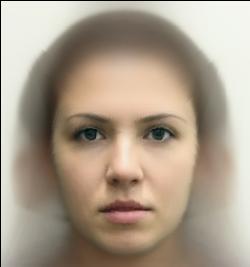}
\caption {Examples of faces: first row shows real images in the training set, second row shows the corresponding reconstructions of the real faces and the third row shows synthetic faces sampled from the training face-space.} \label{fig:face_samples}}
\end{figure*}
AMs describe the variation contained within the training set of faces, used for its construction. Given that this set spans all variations associated with identity changes, the AM provides a good approximation to any desired face. This approximation is represented by a point (or more precisely, localized contiguous region) within the face-space, defined by the \emph{AM} \emph{coefficients}. The distribution of AM coefficients of faces belonging to the same ethnicity are well approximated by an independent, multivariate, Gaussian probability density function~\cite{GibsonSB03,Matthews_2004_4601,TzimiropoulosP13,WuLD08} (for example, see Figure~\ref{fig:distr_ex} that presents the distribution of the first 21 AM coefficients for a face-space constructed from 500 faces.) New instances of facial appearance, the synthetic faces, can be obtained by randomly sampling from such a distribution.  For simplicity, hereafter we assume that faces belong to a single ethnicity. To accommodate faces from different ethnic backgrounds, the same concept could be used with the mixture of Gaussians distribution.

We follow the procedure for AM construction, proposed in~\cite{GibsonSB03}. The training set of facial images, taken under the same viewing conditions, is annotated using a point model that delineates the face shape and the internal facial features. In this process, 22 landmarks are manually placed on each facial image. Based on these points, 190 points of the complete model are determined (see~\cite{GibsonSB03} for details). For each face, landmark coordinates are concatenated to form a shape vector, $x$. The data is then centered by subtracting the mean face shape, $\bar{x}$, from each observation. The shape principle components $P_s$ are derived from the set of mean subtracted observations (arranged as columns) using PCA.
The synthesis of a face shape (denoted by $\hat{x}$) from the {\em shape model} is done as follows,
\begin{equation}\label{eq:shape_bit}
\hat{x}=P_s b_s + \bar{x},
\end{equation}
where $b_s$ is a vector in which the first $m$ elements are normally distributed parameters that determine the linear combination of shape principal components
and the remaining elements are equal to zero. We refer to $b_s$ as the {\em shape coefficients}.

Before deriving the texture component of the AM, training images must be put into correspondence using non-rigid shape alignment procedure.
Each shape normalized and centered RGB image of a training face is then rearranged as a vector $g$. Such vectors for all training faces form a matrix which is used to compute the texture principle components $P_g$ by applying PCA.
A face texture  (denoted by  $\hat{g}$) is reconstructed from the {\em texture model} as follows,
\begin{equation}\label{eq:texture_bit}
\hat{g}=P_g b_g + \bar{g},
\end{equation}
where $b_g$ are the {\em texture coefficients} which are also normally distributed and $\bar{g}$ is the mean texture.

The final model is obtained by a PCA on the concatenated shape and texture parameter vectors. Let $Q$ denote the principal components of the
concatenated space. The AM coefficients ($c$)
are obtained from the corresponding shape ($x$) and texture ($g$) as follows,
\begin{equation}\label{eq:AM}
c=Q^T\left[
          \begin{array}{c}
             r b_s\\
             b_g\\
          \end{array}
        \right]\equiv Q^T\left[
          \begin{array}{c}
             w P_s^T(x-\bar{x})\\
             P_g^T(g-\bar{g})\\
          \end{array}
        \right]
\end{equation}
where $w$ is a scalar that determines the weight of shape relative to texture.

Figure~\ref{fig_AMM_example} illustrates the shape vector $x$ (center image) and the shape free texture vector $g$ (on the right) used to obtain the AM coefficients.

AM coefficients of a \textbf{real face} are obtained by projecting its shape $x$ and texture $g$ onto the shape and texture principal components correspondingly and then combining the shape and texture parameters into a single vector and projecting it onto the AM principal components.  In order to create the \textbf{synthetic faces}, we first estimate a $d$-dimensional Gaussian distribution ${\cal N}(0^d,\vec{\sigma}^2)$ of the AM coefficients using the training set of real faces.  
Then AM coefficients of synthetic faces are obtained by directly sampling from this distribution, discarding the samples beyond $s$ standard deviations.\footnote{We chose $s$ such that all training samples are within $s$ standard deviations from the mean.}

Theoretically, the expected distance between the samples from AM distribution to its center is
about $\sqrt{d}$ standard deviation units. We observed that the distance of real faces from the center is
indeed close to $\sqrt{d}$ standard deviation units. In other words, AM coefficients are most likely to
lie on the surface of an $d$-dimensional ellipsoid with radii $r_i=k\cdot \sigma_i$, where $k\approx \sqrt{d}$.
Hence to sample synthetic faces, we use the following process: sample $v$ from a $d$-dimensional Gaussian ${\cal N}(0^d,1^d)$, normalize $v$ to the unit length and multiply coordinate-wise by $k\vec{\sigma}$. To minimize the differences between the AM representations of real and synthetic faces, we apply the same normalization process to the AM coefficients of the real faces as well.



\subsection{Reconstruction}\label{sec:sub:recon}

The biometric ``password file'' of the HoneyFaces system is composed of records, containing the
AM coefficients of either real or synthetic faces. The coefficients are
sufficient for the authentication process without reconstructing the face.  However, we use reconstructed faces in our privacy and security analysis, thus in the following, we show how to reconstruct faces from their corresponding AM coefficients.

First, the shape and texture coefficients are obtained from the AM coefficients as follows,
$b_s = Q_s c$ and $b_g = Q_g c$, where $[Q_s^T Q_g^T]^T$ is the AM basis. Then the texture and shape of the face are obtained via Eq.~(\ref{eq:shape_bit})
and~(\ref{eq:texture_bit}).
Finally, the texture $\hat{g}$ is warped onto the shape $\hat{x}$, resulting in a facial image. Figure~\ref{fig:face_samples} shows several examples of reconstructed real faces and synthetic faces, sampled from the estimated distribution of AM coefficients.


\section{Implementation of HoneyFaces Concept} 
\label{sec:System}
\begin{figure*}
\center{
   \includegraphics[height=1.8in]{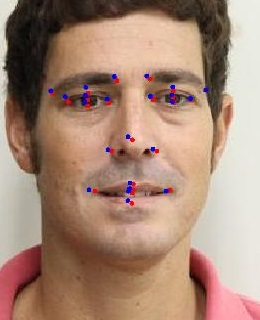}
  \includegraphics[height=1.8in]{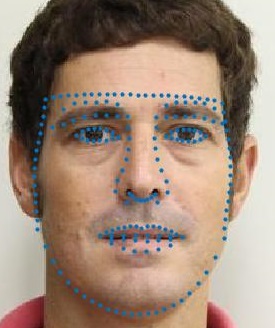}
  \includegraphics[height=1.8in]{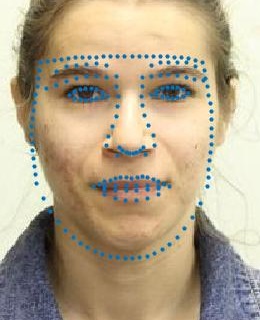}
\caption{The leftmost image shows the landmarks obtained by running the Face++ algorithm on both the reference image (red points) and the test image (blue points) drawn on the reference image. The center image shows the reference face shape transformed to the test image of the same subject  (a legitimate verification attempt). The rightmost image shows the reference shape transformed to a different subject (an imposter verification attempt). The transformed reference shape together with the underlying image is used for the computation of the AM coefficients of the test image.}\label{fig:lmk_samples}}
\end{figure*}
To prevent exfiltration and protect privacy of the users, we create a very large number of synthetic faces. These faces can be incorporated in the authentication system in different ways.
For example, the HoneyWords \cite{HoneyWords} method stores a list of passwords (one of which is correct and the rest are fake) per account. In our settings, both the number of synthetic faces and the ratio of synthetic to real faces should be large. Thus, the configuration, in which the accounts are created solely for real users, requires a very large number of synthetic faces to be attached to each account.  Hence, in such an implementation, during the authentication process, a user's face needs to be compared to a long list of candidates (all fake faces stored with the user name). This would increase the authentication time by a factor equal to the synthetic-to-real ratio, negatively affecting the usability of the system and leading to an undesirable trade off between privacy and usability.

Another alternative is creating many fake accounts with a single face as a password. This does not change the authentication time of the system (compared to a system with no fake accounts). Since most real systems have very regular user names (e.g., the first letter of the given name followed by the family name), it is quite easy to generate fake accounts following such a convention. As we show in Section~\ref{sec:sub:Blowup}, this allows inflating the password file to more than 56.6 TB (when disregarding the storage of user names).

One can also consider a different configuration, aimed to fool an adversary
that knows the correct user names, but not the real biometrics. Specifically, we can store several faces in each account (instead of only one) in addition to the fake accounts (aimed at adversaries without knowledge of user names). The faces associated with a fake account are all synthetic. The faces associated with a real account include one real face of that user and the rest are synthetic one. In such a configuration the authentication time does not increase significantly, but the total size of the ``biometric data'' and the ratio of real-to-synthetic faces remains large. Moreover, the adversary that knows the user name still needs to identify the real face among several synthetic faces.

In this work we implemented and analyzed the configuration in which real and decoy users have an account with a single password. The majority of the users are fake in order to hide the real ones. Each user (both real and fake) has an account name and a password composed of 80 AM coefficients.  These coefficients are derived from the supplied facial image for real users or artificially generated for decoy ones.

\subsection{Face Modeling}
The number of training subjects for the face-space construction should be larger than the number of system users. This provides a better modeling of the facial appearance, allowing a large number of synthetic faces to be created, and protecting the privacy of system's users as discussed in Section~\ref{sec:PrivacyAnalysis}.
We used a set of 500 subjects to train the AM. 270 of them were the users of the HoneyFaces system. All images in the training set were marked with manual landmarks using a tool similar to the one used in~\cite{GibsonSB03}. We computed a 50-dimensional shape model and a 350-dimensional texture model as described in Section~\ref{sec:sub:Generating} and we reduced the dimension of the AM parameters to 80. We note that this training phase is done once, and needs to contain the users of the system mainly for optimal authentication rates. However, as we later discuss in Section~\ref{sec:sub:FaceSpace}, extracting biometric information of real users from the face-space  is infeasible.

All representations used in the system were normalized to unit norm and then multiplied by 7 standard deviations.  This way we forced all samples (real and synthetic) to have the same norm, making the distribution of distances of real and synthetic faces very similar to each other (see Figure~\ref{fig:distr_ex}).

\subsection{Inflating the Biometric ``Password File'' with Synthetic Faces}
\label{sec:sub:Blowup}
We used the resulting face-space to generate synthetic faces. We discarded synthetic faces that fall closer than a certain distance from the real or previously created synthetic faces. The threshold on the distance between faces of different identities was set to 4,800, thereby minimizing the discrepancy between the distance distributions of real and synthetic faces. This minimum separation distance prevents collisions between faces and thus the addition of synthetic faces does not affect the authentication accuracy of the original system (prior to inflation).

The process of synthetic face generation is very efficient and takes only  $1.2903\times10^{-4}$ seconds on average using Matlab.

Using 500 training faces we were able to create about $2^{36.5}$ synthetic faces, with sufficient distance from each other. We strongly believe that more faces can be generated (especially if the size of the training set is increased), but $2^{36.5}$ faces that occupy 56.6 TB seems sufficient for proof of concept.


\begin{figure}
\center{
   \includegraphics[height=2.3in]{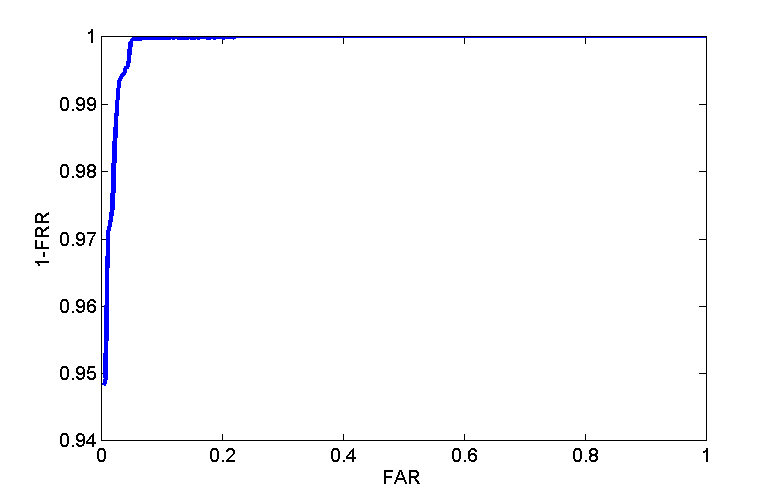}
\caption{ROC curve of the verification experiment, specified in Section~\ref{sec:sub:Usbaility}.}\label{fig:ROC}}
\end{figure}

\subsection{Authentication}

The authentication  process of most biometric systems is composed of the user
supplying the user name and her or his facial image. This image (hereafter the test image)
is aligned to conform with the reference image stored for that user. After the registration, the distance between the test and reference templates are computed and compared to some predefined threshold.

To find the registration between the test and the reference facial templates in our system, we first reconstruct the facial shape of the corresponding subject in the database from the AM coefficients (as shown in Section~\ref{sec:sub:recon}). We then run an automatic landmark detector on the test image (using Face++ landmark detector~\cite{Face++}) and use these landmarks and the corresponding locations in the reference shape to find the scaling, rotation, and translation transformations between them. Then we apply this transformation to the reference shape to put it into correspondence with the coordinate frame of the test image.



The AM coefficients of the test image are computed using the transformed reference shape and the test image itself (as shown in Section~\ref{sec:sub:coefficients}) and then compared to the stored AM coefficients --- the password, using the L2 norm. The threshold on the L2 distance was set to 3,578 which corresponds to 0.01\% of FAR. Note that the threshold is smaller than the distance between the faces (4,800) used for synthetic face generation.  Figure~\ref{fig:lmk_samples} illustrates  the authentication process for the genuine and imposter attempts.



\subsection{Usability as an Authentication System}
\label{sec:sub:Usbaility}

\subsubsection{Accuracy} We ran 270 genuine attempts, comparing the test image with the corresponding reference image, and about 4,200,000 impostor attempts (due to the access to Face++). For a threshold producing an FAR of 0.01\%, our system showed the true acceptance rate (100-FRR)  of 93.33\%. Figure~\ref{fig:ROC} shows the corresponding ROC curve. Our tests showed no degradation in FRR/FAR after the inclusion of synthetic faces.

\subsubsection{Efficiency}
Finding landmarks in a test image, using the Face++ landmark detector~\cite{Face++}, takes 1.42 seconds per subject on average.  We note that the implementation of the landmark detector is kept at the Face++ server and thus the reported times include network communications. Running the detector locally will significantly reduce the running time. Obtaining the AM coefficients  of a test image and comparing them to those of the target identity in the database takes additional 0.53 seconds on average.  This brings us to a total of 1.95 seconds (on average) for a verification attempt.

The system was implemented and tested in Matlab R2014b 64-bit on Windows 7, in 64-bit OS environment with Intel's i7-4790 3.60GHZ CPU and 16GB RAM. Local
implementation that uses C is expected to improve the running times significantly (though not faster than 1 ms).


\section{Privacy Analysis}
\label{sec:PrivacyAnalysis}

Our privacy analysis targets an adversary with access to the inflated
biometric ``password file'', and is divided into three cases. The first scenario, discussed
in Section~\ref{sec:sub:NoPrior}, is an adversary
who has no prior knowledge about the users of the system. Such an adversary
tries to identify the real users out of the fake ones. The second
scenario, discussed in Section~\ref{sec:sub:out_source}, concerns an
adversary that tries to achieve the same goal, but has access to a comprehensive, external source of facial images that adequately represents the world wide variation (population) in facial appearance but does not know who the users are.
The last scenario assumes that the adversary
obtained the biometric data of all but one out of the system's users,
and wishes to use this
to find the biometrics of the remaining user. We discuss this case in Section~\ref{sec:sub:FaceSpace}.

\subsection{Privacy with No Prior Knowledge}
\label{sec:sub:NoPrior}

We first discuss the scenario in which the adversary has the full database
(e.g., after breaking into the system) and wishes to identify the real users but has
no prior knowledge concerning the
real users. More explicitly, this assumption means that the adversary does
not have a candidate list and their biometrics, to check if they are in the database.
\begin{definition}
An inflated password file is a file that contains $N$ facial templates, $n$ of which correspond to real faces and remaining $N-n$ are synthetic faces sampled from the same face-space as the real faces.
\end{definition}
\begin{definition}
A simulated password file is a file that contains $N$ facial templates, all of which are  synthetic faces sampled from the same face-space.
\end{definition}

\begin{lemma}
An adversary that can distinguish between an inflated password file and a
simulated password file, can be transformed into an adversary that extracts
all the real users. Similarly, an adversary that can extract real users from
a password file can be used for distinguishing between inflated
and simulated password files.
\end{lemma}

\begin{proof}
We start with the simpler case --- transforming an adversary that
can extract the real faces into a distinguisher between the two files. The
reduction is quite simple. If the adversary can extract real faces
out of the password file (and even only a single real face), we just give
it the password file we have received. If the adversary succeeds in
extracting any face out of it, we conclude that we received the inflated password
file. Otherwise, we conclude that we received a simulated password file.
It is easy to see that the running time of the distinguishing attack and
its success rate are exactly the same as that of the original extraction
adversary.

Now, assume that we are given an adversary that can distinguish between an inflated password file and a simulated one with probability $\epsilon$. We start by recalling that
the advantage of distinguishing between two simulated ones is necessarily zero. Hence,
one can generate a hybrid argument, of replacing one face at a time in the
file. When we replace a synthetic face with a different synthetic face, we have not
changed the distribution of the file. Thus, the advantage drops only when we replace
a real face with a synthetic face, which suggests that if there are
$n$ real users in the system, and $N$ total users in the system, we can
succeed in identifying at least one of the real users of the system with probability
greater than or equal to $\epsilon/n$ and running time of at most $N$ times the running time
of the distinguishing adversary.
\end{proof}

\begin{corollary}
\label{cor_1}
If the distributions of the inflated password file and the simulated password file
are statistically indistinguishable, an adversary
with no prior knowledge (of either user's biometrics or user names) cannot identify the real users.
\end{corollary}



Theoretically, synthetic and real faces are sampled from the same distribution and thus are indistinguishable according to Corollary~\ref{cor_1}. However, in practice, synthetic faces are sampled from a parametric distribution which is estimated from real faces. The larger the set of faces, used to estimate  the distribution, the closer these distributions will be. In practice, the number of training faces is limited which could introduce some deviations between the distributions. Our following analysis shows that these deviations are too small to distinguish between the distributions of real and synthetic faces either by statistical tests or by human observers.


The first part of the analysis performs a statistical test of the AM coefficients of the real and the synthetic faces and shows that these distributions are indeed very close to each other. The second part studies the distribution of mutual distances among real and synthetic faces and reaches the same conclusion. Finally, we perform a human experiment on the reconstructed and simulated faces, showing that even humans can not distinguish between them.

\subsubsection{Statistical Tests on the AM Coefficients}
\label{sec:AM_coeffs_attack}

\begin{figure*}[t]
\includegraphics[width=8in]{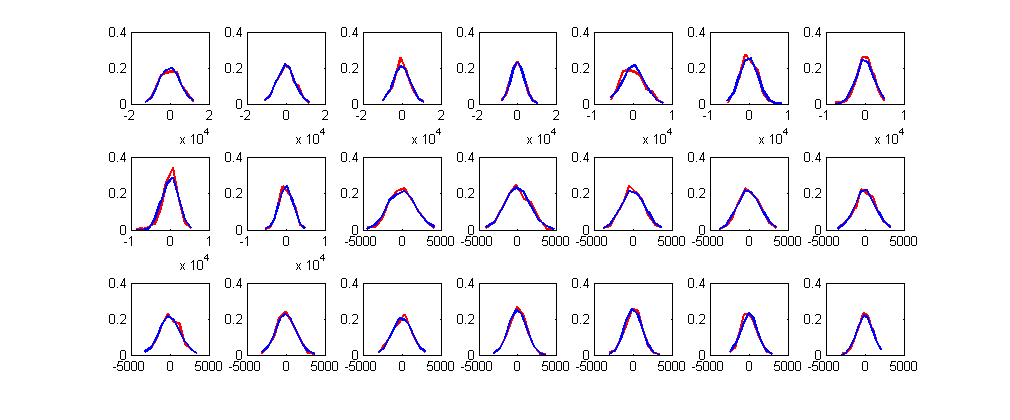}
\caption {Comparing the distributions of the AM coefficients for the first 21 (out of 80) principle components. Real face are shown in red and synthetic ones in blue.} \label{fig:distr_ex}
\end{figure*}

The AM coefficients are well approximated by a Gaussian distribution in all dimensions~\cite{GibsonSB03,Matthews_2004_4601,TzimiropoulosP13,WuLD08}.
Therefore, sampling AM coefficients for synthetic faces from the corresponding distribution is likely to produce representations that cannot be distinguished by standard hypothesis testing from real identities.  The examples of real and synthetic distributions for the first 21 dimensions are depicted in Figure~\ref{fig:distr_ex} and the following analysis verifies this statement.

First, we show that coefficients of real and synthetic faces cannot be reliably distinguished based on two sample Kol\-mogo\-rov--Smirnov (KS) test. To this end, we sampled a subset of 500 synthetic samples from 80-dimensional AM and we compare it to the 500 vectors of coefficients of training images. We ran the KS test on these two sets for each of the 80 dimensions and recorded the result of the hypothesis test and the corresponding p-value. We repeated this test 50 times, varying the set of synthetic faces. The KS tests supported the hypothesis that the two samples come from the same distributions in 98.72\% of the cases with a mean p-value 0.6 (over 50 runs and 80 components, i.e., 4000 tests). These results show that AM coefficients of real and synthetic faces are indistinguishable using a two-sample statistical test.

\subsubsection{Distribution of Mutual Distances}\label{sec:distance_attack}

We analyzed the distributions of distances between
the real faces, synthetic ones, and a mixture of both.
Figure~\ref{fig:dist_distr} shows that
these distributions, both in the case of Euclidean distances and in the case of angular distances,\footnote{Angular distance between the faces is defined to be one minus the cosine of the included angle between points (treated as vectors).} are very close. Hence, the statistical distance between them is negligible,
suggesting that attacks trying to use mutual distances
are expected to be ineffective.

\begin{figure*}[t]
\center{
\includegraphics[height=2in]{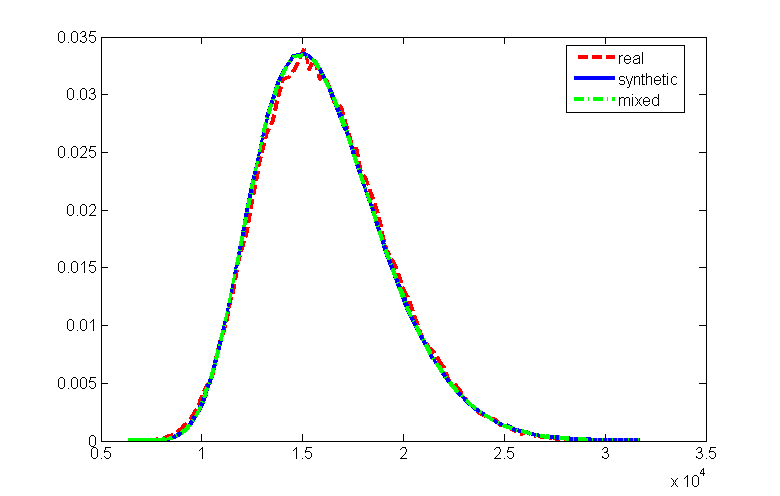}
\includegraphics[height=2in]{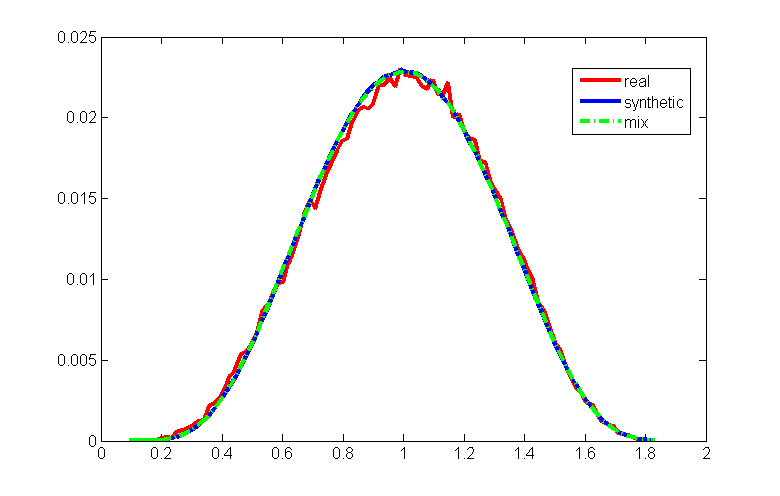}
\caption {Distributions of distances between different real faces, different synthetic faces, and a mix of them. On the left Euclidean distance, and on the right angular distance.} \label{fig:dist_distr}}
\end{figure*}

\subsubsection{Human Experiment}
\label{sec:HumanDistinguisher}


We conducted a human experiment containing two steps. In the first step,
the participants were shown
a real face, not used in the experiment, and its reconstruction.\footnote{The
real face was decomposed into AM coefficients and then reconstructed.}
In the second step of the experiment, each participant was presented with
the same set of 16 faces (11 of which were synthetic and 5 of which were real)
and was asked to classify them as real or fake.
We also allowed the users to avoid
answering in the case of uncertainty or fatigue.

The 11 synthetic faces
were chosen at random from all the $2^{36.5}$ synthetic faces we generated,
and the 5 real ones were chosen at random from the 500 real faces.
For the real faces, we computed the AM coefficients for each real image and then used
the method described in Section~\ref{sec:sub:recon} to generate the real faces and synthetic faces from the model.
Examples of real and synthetic faces are provided in the second and third rows, respectively, of Figure~\ref{fig:face_samples}.

Out of 179 answers we have received, 97 were correct, showing a success rate
of~54.19\%. The fake faces received 120 answers, of which 66 were correct
(55\%). The real faces received 59 answers, of which 31 were correct
(52.5\%). Our analysis shows that the answers for each face are distributed
very similarly to the outcome of a Binomial random variable with a probability of success at each trial of 0.5.



\subsection{Adversary with an External Source of Facial Images}
\label{sec:sub:out_source}
Next, we analyze the case where an adversary has access to the inflated ``password file'' and to an extensive external source of facial images (e.g. the Internet). We consider two attack vectors: the first tries to use  membership queries with random facial images to match real users of the system, the second attempts to distinguish  between real and synthetic faces using a training process on a set of real facial images unrelated to the users of the system.

\subsubsection{Membership Queries}
\label{sec:sub:Membership}

An adversary could use a different source of facial images to try and run a membership query against the HoneyFaces system to obtain the biometric of the real users. To match a random image from an external source of facial images, the adversary must run the authentication attempt with all users of the system (including the fake ones). Our experiments show that the current implementation takes about 2 seconds per authentication attempt (mostly due to the landmarking via Face++).
Even under the unrealistic assumption that the authentication time could be reduced to 1 ms, it would take about $2^{36} \cdot 0.001 = 2^{26}$ seconds (slightly more than 2 CPU years),
to run the matching of a single facial image against $2^{36}$ fake faces. We
note that one cannot use a technique to speed up this search and comparison (such
as kd-trees) as the process of comparison of faces requires aligning them (based
on the landmarks), which cannot be optimized (to the best of our knowledge).


One can try to identify the membership of a person in the system by projecting his/her image onto the face-space of the system and analyzing the distance from the projection to the image itself. If the face-space was constructed from system users only, a small distance could reveal the presence of the person in the face-space.  Such an attack can be easily avoided by building the face-space from  a sufficiently large (external) source of faces. Such a face-space approximates many different appearances (all combinations of people in the training set) and thus people unrelated to the users of the system will also be close to the face-space. We conclude that a membership attack to obtain the real faces from the data base is impractical.

\subsubsection{Machine Learning Attacks}
\label{sec:ML_attacks}

The task of the adversary who obtained the inflated ``biometric password file'' is to distinguish the
real faces from the synthetic ones. He can consider using a classifier that was trained to separate real faces from the fake ones. To this end the adversary needs to construct a training set of real and synthetic faces. Synthetic faces can be generated using the system's face-space. However, the real faces of the system  are unavailable to the adversary.

One way an adversary might approach this problem is by employing a different set of real faces (a substitute set) to construct the face-space. He then can create a training set by generating synthetic faces using that space and reconstructing the real faces from the substitute set following the algorithms described in Section~\ref{sec:sub:Generating}. A trained classifier could then be used to classify the faces in the biometric ``password file''.

The substitute training set is likely to have different characteristics than the original one. The adversary could try to combine the system's face-space with the substitute set in attempt to improve the similarity of the training set to the biometric ``password file''.
%
Then, the adversary can construct the training set of real faces by projecting the images from the substitute set on the mixed face-space and reconstructing them as described in Section~\ref{sec:sub:recon}. To create a training set of synthetic faces, the adversary can either use the mixed face-space or the system's face-space.

Deep learning and, in particular, convolutional neural networks (CNN), showed close to human performance in face verification~\cite{Parkhi15,TaigmanYRW15}. It is a common believe that the success of the CNN in recognition tasks is due to its ability to extract good features. Moreover, it was shown that CNN features can be successfully transferred to perform  recognition tasks in similar domains (e.g.~\cite{Donahue_ICML2014,ROSS,TaigmanYRW15}). Such techniques are referred to as fine tuning or transfer learning. It proceeds by replacing the upper layers of the fully trained DL network (that solves a related classification problem) by layers that fit the new recognition problem (the new layers are initialized randomly).  The updated network is then trained on the new classification problem with the smaller data set. Note, that most of the network does not require training, only slight tuning to fit the new classification task, and the last layer can be well trained using good CNN features and a smaller data set.

Following this strategy, we took the VGG-face deep network \cite{Parkhi15} that was trained to recognize 2,622 subjects and applied the transfer learning method to train a DL network to classify between real and synthetic faces. To this end, we replaced the last fully connected 2,622 size layer with a fully connected layer of size 2 and trained this new architecture in the following settings.

In all experiments we split the training set for training and validation of the network. Then we applied the trained network on a subset of system's data set to classify the images into real and synthetic. The subset included all real faces and a subset of the synthetic faces (same size as real set to balance the classification results).

Setting 1: A face-space was constructed from 500 faces belonging to the substitute set. The training set included 400 reconstructed real faces and  400 synthetic faces, generated  using the substitute face-space. The validation set included 100 reconstructed real faces and 100 synthetic faces from the same domain, not included in the training set. The results on the substitute validation set showed that the DL network classifies 62.5\% of faces correctly. The results on the system's set dropped to 53.33\%, which is close to random.

Setting 2:
A face-space was constructed by combining the system's face-space with the substitute set. The training set included 400 real faces projected and reconstructed using the mixed face-space and  400 synthetic faces, generated using the mixed face-space. The validation set included 100 reconstructed real faces and 100 synthetic faces from the same domain, not included in the training set. The results on the validation set showed good classification: 75\% of synthetic faces were classified as synthetic and 93\% of real faces were classified as real. However, the same network classified all faces of the system's face set as synthetic. This result shows that using a mixed face-space to form a training set is not effective. The prime reason for this is the artifacts in synthetic images due to variation in viewing conditions between the sets.

Setting 3: The real training and validation sets were the same as in Setting 2. The synthetic training and validation sets were formed by generating synthetic faces using system's face-space. Here the classifier was able to perfectly classify the validation set, but it classified all system's faces as synthetic. This shows that using real and synthetic faces from different face-spaces introduces even more differences between them, which do not exist in system's biometric ``password file''.

To conclude, the state-of-the-art deep learning classifier showed accuracy of 53.33\% in distinguishing between  real and synthetic faces in the system's biometric ``password file''. This result is close to random guessing.





\subsection{Finding the Last User (or what can you Learn from the Face-Space)}
\label{sec:sub:FaceSpace}

An adversary who obtains the facial images of all but one of the real users of the system can try and use it for extracting information about the remaining user from the password file. If the training set used for constructing the face-space contains only the users of the system, the following simple attack will work: Recall that the authentication procedure requires removing the mean face from the facial image obtained in the authentication process. Thus, the mean of all faces in the training set is stored in the system. The adversary can find the last user by computing the mean of the users he holds and solving a simple linear equation. To mitigate this attack, and to allow better modeling of the facial appearance, the training set should contain a significant amount of training faces that are not users of the system. Note that these additional faces must be discarded after the face-space is constructed.

Assuming that the training set for the face-space construction was not limited to a set of system users (as is the case in our implementation), the adversary could try the following attack. Create $K=N-n+1$ face-spaces by adding each unknown face from the biometric ``password file'' in turn to the $n-1$ real faces that are in the possession of the adversary. $K$ is equal to the number of synthetic faces in the biometric ``password file'' plus one real face. Then compare these $K$ face-spaces to the one stored in the system (using statistical distance between the distributions). Such comparison provides a ranking of unknown faces to be the $n$'th real face. If the attack is effective, we expect the face-space including the $n$'th user to be highly ranked (i.e., to appear in a small percentile). However, if the distribution of the rankings associated with the face-space including the $n$'th real face over random splits of $n-1$ known and 1 unknown face is (close to) uniform, then we can conclude that the adversary does not gain any information about the last user using this attack.

In our implementation of the attack, we assume that the adversary knows 269 faces of real users and he tries to identify the last real user among the synthetic ones.  Running the attack with all synthetic faces is time consuming.  To get statistics of rankings we can use a much smaller subset of synthetic faces. Specifically, we used 100 synthetic faces and ran the experiment over 100 randomized splits into 269 known and 1 unknown faces. Figure~\ref{fig:dist_rankings} shows the histogram of rankings associated with the face-space including the last real user in 100 experiments. The histogram confirms that the distribution of rankings is indeed uniform, which renders the attack ineffective.


An alternative approach, that the adversary may take, is to analyze the effects of a single face on the face-space distribution. However, our experiments show that the statistical distances
between neighboring distributions (i.e., generated from
training sets that differ by a single face) are insignificant. Specifically, the average statistical distance between the distribution estimated from the full training set (of 500 real faces) and all possible sets of 499 faces (forming 500 neighboring sets, each composed of a different
subset of 499 faces) is $1.2809\cdot 10^{-5}$ and the maximal distance is $6.3696\cdot 10^{-5}$. These distances are negligible compared to the standard deviations of the face-space Gaussians (the largest std is 6,473.7 and the smallest is 304.1717).  These small differences suggest that one can use differential privacy mechanisms with no (or marginal) usability loss (for example, by using ideas related to~\cite{DPPCA}) to mitigate attacks that rely on prior knowledge of the system's users. We leave the implementation and evaluation of this mechanism for future research.

\begin{figure*}[t]
\center{
\includegraphics[height=2in]{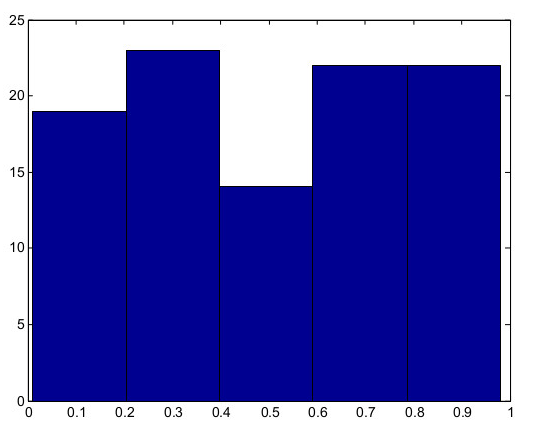}
\caption {Histograms of rankings associated with the face-space including the last unknown users in comparison to the system's face-space over 100 random choices of the unknown real face. The histogram agrees with the uniform distribution} \label{fig:dist_rankings}}
\end{figure*}

To conclude, the HoneyFaces system protects the privacy of users even in the extreme case when the adversary learned all users but one, assuming that the training set for constructing the face-space contains a sufficiently large set of additional faces.

\section{Security Analysis}
\label{sec:SecurityAnalysis}

We now discuss the various scenarios in which HoneyFaces improve the
security of a biometric data.
We start by discussing the scenario
of limited outgoing bandwidth networks (such as air-gaped networks), and showing
the affects of the increased file size on the exfiltration times. We follow
by discussing the effects HoneyFaces has on the detection of the
exfiltration process. We conclude the analysis of the security
offered by our solution in the scenario of partial exposure of the
database.


\subsection{Exfiltration Times in Limited Outgoing Bandwidth}

The time needed to exfiltrate a file is easily
determined by the size
of the file to be exfiltrated and the bandwidth. When the exfiltration
bandwidth is very slow (e.g., in the air-gap networks studied
in~\cite{Shamir1,Shamir2}), a 640-byte representation of a face (or
5,120-bit one) takes between 5 seconds (at 1,000 bits per second rate) to
51 seconds (in the more realistic 100 bits per second rate). Hence,
leaking even a 1 GByte database takes between 92.6 to 926 days (assuming
full bandwidth, and no need for synchronization or error correction overhead).
The size of the password file can be inflated to contain all the
$2^{36.5}$ faces we created, resulting in a 56.6 TBytes file size
(whose leakage would take about 14,350 years in the faster speed).

A possible way to decrease the file size is to compress the file.
Our experiments show that Linux's
zip version 3.0, could squeeze the password file by only 4\%. It is highly unlikely
that one could devise a compression algorithm that succeeds in compressing
significantly more. In other words, compressing the face file reduces
the number of days to exfiltrate 1 GByte to 88.9 days (in the faster
speed).

One can consider a lossy compression algorithm, for example by using only the coefficients associated with the  principle components (carrying most information). We show in Section~\ref{sec:sub:Rates} that this approach requires using many coefficients for identification. Hence, we conclude that if the bandwidth is limited, exfiltration of the full database in acceptable time limit is infeasible.

\subsection{Improved Leakage Detection}
The improved leakage detection stems from two possible defenses: The use
of Intrusion Detection Systems (and Data Loss Prevention products) and the
use of a two-server settings as in HoneyWords.

Intrusion detection systems, such as snort, monitor the network for
suspicious activities. For example, a high outgoing rate of DNS queries
may suggest an exfiltration attempt and raise an alarm~\cite{IDS-story}.\footnote{The analysis reported in~\cite{IDS-story} suggests that an increased outgoing DNS queries at the rate of a few dozens a second was deemed suspicious. The size of a normal DNS query is up to 512 bytes, suggesting that 100 DNS queries per
second can carry 51,200 bytes of information, i.e., a communication rate of
at most 409,600 bits/sec.}
Similar exfiltration attempts can also increase the detection of
data leakage (such as an end machine which changes its HTTP footprint and
starts sending a large amount of information to some external server).
Hence, an adversary who does not take these tools into account is very
likely to get caught. On the other hand, an adversary who tries to
``lay low'' is expected to have a reduced exfiltration rate, preventing
quick leakage and returning to the scenario discussed in the previous
section.

The use of HoneyFaces also allows for a two-server authentication setting
similarly to HoneyWords~\cite{HoneyWords}. The first server uses a database
composed of the real and the synthetic faces. After a successful login
attempt is made into this system, a second authentication query is
sent to the second server, which holds only the real users of the system.
A successful authentication to the first server that uses a fake
account, is thus detected at the second server, raising an
alarm.

\subsection{Analyzing Partial Leakage}
\label{sec:sub:Rates}
\begin{figure*}[t]
\center{
\includegraphics[height=1.7in]{face_real_1_1.jpg}
\includegraphics[height=1.7in]{face_1_1.jpg}
\includegraphics[height=1.7in]{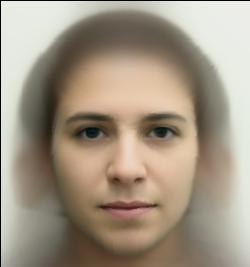}
\includegraphics[height=1.7in]{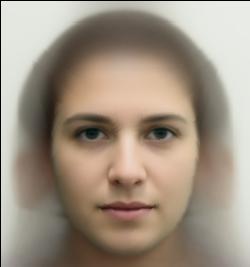}\\

\includegraphics[height=1.7in]{face_real_3_1.jpg}
\includegraphics[height=1.7in]{face_3_1.jpg}
\includegraphics[height=1.7in]{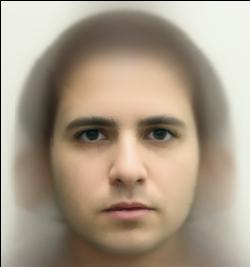}
\includegraphics[height=1.7in]{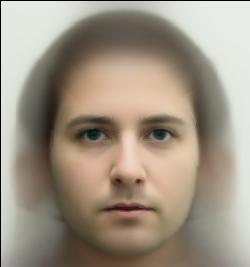}\\

\includegraphics[height=1.7in]{face_real_11_1.jpg}
\includegraphics[height=1.7in]{face_11_1.jpg}
\includegraphics[height=1.7in]{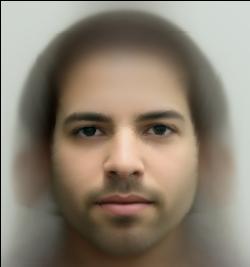}
\includegraphics[height=1.7in]{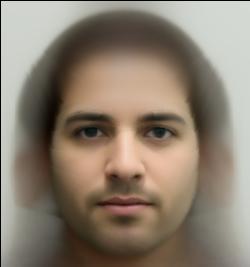}\\
\caption {Examples of reconstructions: from left to right 80, 30, 10 coefficients. The resemblance with the original image declines with the number of coefficients used for reconstruction. Reconstruction with less than 30 coefficients has no similarity to the original image. } \label{fig:grad_rec}}
\end{figure*}
We showed that exfiltrating the entire password file in acceptable time is infeasible if the bandwidth is limited. Hence, the adversary can decide to pick one of two approaches (or combine
them) when trying to exfiltrate the file --- either leak only partial
database (possibly with an improved ratio of real to synthetic faces), or to leak partial representations such as the first 10 AM
coefficients out of 80 per user.

As we showed in the privacy analysis (Section~\ref{sec:PrivacyAnalysis}), statistical tests or Machine Learning methods fail to identify the real faces among the synthetic ones. Using membership queries to find real faces in the database is computationally infeasible without prior knowledge of the real user names.
We conclude that reducing the size of the data set by identifying the real users or significantly improving the real to synthetic ratio is impossible.

The second option is to leak a smaller number of the coefficients (a partial representation). Leaking a smaller number of coefficients can be done faster than the entire record,
and allow the adversary to run on his system (possibly with greater computational power), any algorithm he wishes for the identification of the real users.
In the following, we show that partial representations (that significantly decrease the size of the data set) do not provide enough information for successful membership queries.

We experimented with 10 coefficients (i.e., assume that the adversary leaked
the first 10 AM coefficients of all users). As the adversary does not
know the actual threshold for 10 coefficients, he can try and approximate
this value using the database. Our proposed method for this estimation is
based on computing the distance
distribution for 30,000 faces from the database, and setting a threshold
for authentication corresponding to the 0.01\% ``percentile'' of
the mutual distances. We then take
test sets of real users' faces and of outsiders' faces
and for each face from these sets, computed the minimal distance
from this face to all the faces in the reduced biometric ''password file''. We assume that if this
distance is smaller than the threshold, then the face was in the system,
otherwise we conclude that the face was not in it.

Our experiments show that for the 0.01\% threshold, 98.90\% of the outsider
set and 99.26\% of the real users were below the threshold. In other words,
there is almost no difference between the chance of determining that a user of
the system is indeed a user vs.~determining that an outsider is a user of the
system. This supports the claim that 10 coefficients are insufficient to
distinguish between real users and outsiders.

We also used a smaller threshold which
tried to maximize the success rate of an outsider to successfully match to
a real face. For this smaller threshold, 71.08\% of the outsiders were
below it compared with 74.07\% of the real users.

To further illustrate the effects of partial representation on the reconstructed
face, we show in Figure~\ref{fig:grad_rec} the reconstruction of faces from
80, 30, and 10 coefficients, compared with the real face. As can be seen,
faces reconstructed from 30 coefficients are somewhat related to the original
face, but faces reconstructed from 10, bare little resemblance to the original.
Although it is possible to match a degraded face to the corresponding original when a small number of faces are shown (Figure~\ref{fig:grad_rec}), visual
matching is impossible among $2^{36.5}$ faces.

Thus, an adversary wishing to leak partial information about an
image, needs to leak more than 10 coefficients.

To conclude, exfiltrating even a partial set of faces (or parts of the
records) does not constitute a plausible attack vector against
the HoneyFaces system.

\section{Summary}
\label{sec:Summary}

In this paper we explored the use of synthetic faces for increasing the security and privacy of face-based authentication schemes. We have proposed a
new mechanism for inflating the database of users (HoneyFaces) which guarantees users' privacy with no usability loss. Furthermore, HoneyFaces offers improved resilience against
exfiltration (both the exfiltration itself and its detection). We also showed that this mechanism does not interfere with the basic authentication role of the system
and that the idea allows the introduction of a two-server authentication solution as in
HoneyWords.

Future work can explore the application of the HoneyFaces idea to other
biometric traits (such as iris and fingerprints). We believe that due to the
similar nature of iris codes (that also follow multi-dimensional Gaussian
distribution), the application of the concept is going to be quite straightforward.

\section*{Acknowledgments}

The funds received under the binational
UK Engineering and Physical Sciences Research Council
project EP/M013375/1 and Israeli Ministry of Science and
Technology project 3-11858,
``Improving cyber security using realistic synthetic face
generation'' allowed this work to be carried out.

\bibliographystyle{abbrv}

\end{document}